\documentclass[a4paper,superscriptaddress,aps,nofootinbib,longbibliography,twocolumn]{revtex4-2}
\pdfoutput=1

\usepackage{graphicx} % Required for inserting images
\usepackage{amsfonts}
\usepackage{amsmath}
\usepackage{amssymb}
\usepackage{comment}
\usepackage{amsthm}
\usepackage[normalem]{ulem}
\usepackage[colorlinks,citecolor=blue,linkcolor=blue,urlcolor=blue, breaklinks=true]{hyperref}
\newcommand{\K}{\mathcal K}
\newcommand{\J}{\mathcal J}
\newcommand{\I}{\mathcal I}

\newcommand{\ket}[1]{{| #1 \rangle}}
\newcommand{\bra}[1]{{\langle #1 |}}

\DeclareMathOperator{\tr}{tr}

\usepackage{xcolor}
\usepackage{xspace}
\usepackage{ifthen}
\newcommand{\showcomments}{true}

\newcommand{\rloca}[1]%
{\ifthenelse{\equal{\showcomments}{true}}{{\color{green}{[#1]}}}{\xspace}}%

\definecolor{navy}{RGB}{0,70,180}

\newcommand{\andrea}[1]%
{\ifthenelse{\equal{\showcomments}{true}}{{\color{navy}{[#1]}}}{\xspace}}%

\newtheorem{definition}{Definition}[]
\newtheorem{theorem}{Theorem}[]

\usepackage{ifthen}
\usepackage[normalem]{ulem}
\newcommand{\trackchanges}{true} %% change to true/false 

\definecolor{darkgreen}{RGB}{0,128,42}
\definecolor{darkred}{RGB}{179,0,0}

\newcommand{\removed}[1]{%
\ifthenelse{\equal{\trackchanges}{true}}%
{{\small{\color{darkred}{\sout{#1}}}}}%
{\xspace}%
}

\newcommand{\added}[1]{%
\ifthenelse{\equal{\trackchanges}{true}}%
{{\color{orange}{#1}}}%
{#1}%
}%
\begin{document}

\author{Carlo Cepollaro}
\thanks{These authors contributed equally to this work.}
\affiliation{Vienna Center for Quantum Science and Technology (VCQ), Faculty of Physics, University of Vienna, Boltzmanngasse 5, A-1090 Vienna, Austria}
\affiliation{Institute for Quantum Optics and Quantum Information (IQOQI),
Austrian Academy of Sciences, Boltzmanngasse 3, A-1090 Vienna, Austria}

\author{Andrea Di Biagio}
\thanks{These authors contributed equally to this work.}
\affiliation{Institute for Quantum Optics and Quantum Information (IQOQI),
Austrian Academy of Sciences, Boltzmanngasse 3, A-1090 Vienna, Austria}
\affiliation{Basic Research Community for Physics e.V., Mariannenstraße 89, Leipzig, Germany}

\title{Aumann's theorem beyond ontology:\\quantum, postquantum, and indefinite causal order}

\begin{abstract}
    \noindent Agreement theorems are no--go results about rational disagreement: if two agents start from a common prior and their posterior beliefs are common knowledge, they cannot assign different probabilities to the same event. Standard treatments of the result have the agents reason about an underlying state of the world, which has lead some to ask whether the result can extend to quantum or postquantum phenomena, where such a description may no longer be appropriate. We derive an operational version of Aumann’s agreement theorem without assuming an objective state of the world and instead focusing only on what is observed. This allows us to establish the theorem's validity in quantum theory and even in situations with indefinite causal order or involving hypothetical postquantum phenomena. We comment on seemingly contradictory results in the literature and point to the one place where the theorem might fail: Wigner's friend--type situations. 
\end{abstract}

\maketitle

\section{Introduction}

\noindent While there are general philosophical grounds to question the connection between our observations and the world, the assumption of the existence of an objective reality that is independent of our observations was seen as generally unproblematic for most of the history of science. The development of quantum theory was immediately accompanied with a radical throwing into question of this assumption. Since Heisenberg's seminal paper developing a mechanics “based exclusively on relations between quantities that are, in principle, observable”~\cite{heisenberg1925}, the debate on the relation between our observations and reality in quantum theory continues to this day~\cite{bacciagaluppi2009quantum, wood2025its}. This debate has since given rise to a series of no--go theorems showing with increasing clarity that any realist ontology underlying quantum phenomena must come at the price of a number of deeply unpalatable features~\cite{bell1964einstein,bell1966problem,kochen1967problem,bell1976theory,pusey2012reality,brukner2015quantum,bong2020strong}.

Despite its shaky philosophical foundations, the empirical success of quantum theory has been nothing short of extraordinary. Given enough training, anyone with a description of an experiment would agree on its empirical predictions in the form of expected statistics, even though different interpretations would provide wildly different accounts of what is ``actually'' going on.  Indeed, a large portion of the scientific community advocates that it is meaningless to talk about a reality existing \textit{out there} beyond our observations, because such an idea is unnecessary to account for observed phenomena~\cite{gibney2025physicists,brukner2003information,fuchs2007quantum,bub2008two,fuchs2019qbism,dibiagio2025relative}.

\pagebreak

In a framework where one seeks to avoid commitment to an underlying reality and works instead directly with probabilities for measurement outcomes, it becomes especially important to understand when different agents must arrive at the same probabilistic assignments. It is natural then to revisit classic results from Bayesian epistemics and see if they can be applied to quantum scenarios. A particularly prominent result in this direction is Aumann’s agreement theorem~\cite{aumann1976,geanakoplos1982we,contreras2021observers,brandenburgeragreement,diaz2025quantum,leifer2026generalising}, which states that if rational agents share a common prior and their posterior beliefs are common knowledge, then they cannot assign different probabilities to the same event---they cannot ``agree to disagree,'' so to speak. In its standard formulation, however, the theorem is derived under the assumption of the existence of a true state of the world. 

In this paper we show that this ontological structure is not essential. The result follows directly from the existence of a joint probability distribution over measurement outcomes together with Bayesian conditioning. This operational formulation removes the need for an underlying ontic state space and therefore extends the scope of Aumann’s theorem beyond classical models to arbitrary quantum scenarios, including non--commuting measurements, and even to indefinite causal order and more general postquantum frameworks. What is essential instead is only the existence of a shared probabilistic description of possible observations together with Bayesian conditioning.

We will first review Aumann's agreement theorem. We will then derive our main result---the operational agreement theorem---and show how it applies to very general situations. We will then discuss the relation to previous literature on extensions of Aumann's theorem to quantum and postquantum settings, including seemingly contradictory results. Finally, we will comment on the one place we can expect a failure of the agreement theorem.

\raggedbottom

\section{Classical Aumann's theorem}
\label{sec:aumann}

\noindent Consider two agents, Alice and Bob, who obtain data about the world and, on the basis of their observations, assign posterior probabilities to an event of interest. Aumann's agreement theorem gives a striking conclusion under strong but natural assumptions: if Alice and Bob begin from a \emph{common prior} on the state space and, as a result of their individual measurements, their posterior probabilities are \emph{common knowledge} (that is, Alice knows Bob's posterior, Bob knows Alice's, each knows that the other knows their posterior, and so on ad infinitum), then their posteriors must coincide. This is sometimes described as the fact that two Bayesian agents with a common prior and common knowledge \emph{cannot agree to disagree}. The precise statement is as follows.

Let $\Omega$ be a finite statespace representing the possible states of the world and $p$ be a common prior probability measure on $\Omega$. Alice performs a measurement represented by a partition $\smash{\Pi^{(A)}=\{\Pi_i^{(A)}\}_{i \in \I}}$ and obtains the result $i$, upon which she knows that the true state of the world $\omega^* \in \Omega$ is contained in the $i$--th partition element $\smash{\omega^* \in \Pi^{(A)}_i}$; upon learning her result, she assigns probability $\smash{p_A(E|\Pi^{(A)}_i):={p\big(E \cap \Pi_i^{(A)}\big)}/{p\big(\Pi_i^{(A)}\big)}}$. Bob does an entirely analogous procedure.

To define common knowledge, let $\Pi^{(A)}(\omega)$ be the partition element of $\Pi^{(A)}$ that contains $\omega$. Then
\begin{equation}
\begin{aligned}
p_A(\omega) &:= p\!\left(E \mid \Pi^{(A)}(\omega)\right), \\
p_B(\omega) &:= p\!\left(E \mid \Pi^{(B)}(\omega)\right)
\end{aligned}
\end{equation}
are the posteriors Alice and Bob would assign if the world were in $\omega$.
Fix the posteriors $q_A,q_B \in [0,1]$, and define
\begin{equation}
\begin{aligned}
A_0 &:= \{\omega \in \Omega : p_A(\omega)=q_A\}, \\
B_0 &:= \{\omega \in \Omega : p_B(\omega)=q_B\},
\end{aligned}
\end{equation}
which are the states of the world such that Alice and Bob assign posteriors $q_A$ and $q_B$ to the event $E$. Furthermore, define
\begin{equation}
\begin{aligned}
A_{n+1} &:= \{\omega \in A_n : \Pi_A(\omega)\subseteq B_n\}, \\
B_{n+1} &:= \{\omega \in B_n : \Pi_B(\omega)\subseteq A_n\}.
\end{aligned}
\end{equation}
For example, $A_1$ is the set of the states of the world such that Alice assigns posterior $q_A$ and knows that Bob assigned posterior $q_B$. 

\begin{definition}[Common knowledge of posteriors] The posteriors are common knowledge when Alice knows Bob's posterior, she knows that Bob knows hers, and she knows that Bob knows that she knows his posteriors, and so on indefinitely.
\end{definition}
\noindent In the current setup, common knowledge is achieved whenever $\omega$ is in all $A_n$ and all $B_n$. We'll say that there is \emph{common knowledge at $\omega$} if
\begin{equation}
    \omega \in \bigcap_{n=0}^\infty (A_n \cap B_n).
\end{equation}

Given a prior $p$, partitions $\Pi^{(A)}$ and $\Pi^{(B)}$ and posteriors $q_A$ and $q_B$, it is not guaranteed that such $\omega$ exists; however, 
\begin{theorem}[Aumann's Agreement Theorem]
If Alice's posterior is $q_A$ and Bob's posterior is $q_B$ and this is common knowledge at the true state of the world $\omega^* \in \Omega$,
then $q_A=q_B$.
Equivalently, two Bayesian agents with a common prior cannot have different posterior probabilities for the same event if those posteriors are common knowledge~\emph{\cite{aumann1976}}.
\end{theorem}

Note the logical structure of Aumann’s theorem. The theorem assumes a state space $\Omega$ and that there is a true state of the world $\omega^* \in \Omega$ that is unaffected by Alice and Bob obtaining information about it through their measurements. In this sense, the framework appears to rely on an \emph{ontological} picture: the results of their measurements are determined by the state of the world, and the very notions of information partitions and common knowledge are defined relative to an underlying state space.

One may therefore wonder whether this assumption is essential, or whether an agreement theorem can still be formulated once it is weakened. In fact, such a formulation is indeed possible, and an agreement theorem can be obtained without this strong ontological assumption.  Surprisingly, this substantially extends the domain of its applicability.

\section{Main result}
\label{sec:main}

\noindent We now formulate an agreement theorem which does not rely on a ``true state of the world'' or, more generally, an underlying ontic state space. We do so by restricting our attention to the space of possible measurement outcomes (see Fig.~\ref{fig:figure}).

Alice and Bob do not share a prior on the state of the world, which may or may not exist, but rather on the results of their measurements. This prior probability may be obtained from any theory, including quantum mechanics or a more exotic postquantum theory. The prior, together with Bayesian reasoning, is all that is required to obtain an agreement theorem that applies well--beyond classical theory.

Alice performs a measurement and observes an outcome $i\in\I$, Bob performs another measurement and observes an outcome $j\in\J$, and both agents are interested in the outcome $k\in\K$ of another measurement. We define the space of possible outcome triples as
$\mathcal{M} := \I \times \J \times \K.$
Assume that Alice and Bob share a common prior distribution $p$ on $\mathcal{M}$, i.e. they agree on the probability $p(i,j,k)$ assigned to each triple $(i,j,k)\in\mathcal{M}$ before conditioning on their respective observations.

If Alice observes a measurement outcome $i$, her posterior probability of the outcome $k \in \K$ is calculated using usual Bayesian reasoning, and similarly for Bob:
\begin{equation}
\begin{aligned}
p_A(k|i) := \frac{p(i,k)}{p(i)},~~~~
p_B(k|j) := \frac{p(j,k)}{p(j)},
\end{aligned}
\end{equation}
where probabilities over a subset of measurement results are obtained, as usual, via marginalization, e.g., ${p(i,k)=\sum_{j\in\mathcal J}p(i,j,k)}.$ 

Let $E\subset \mathcal K$ be the event of interest, and $q_A$ and $q_B$ two probabilities. The sets of indices
\begin{equation}
    \begin{aligned}
        A_0&=\{i\in\mathcal I~|~p_A(E|i)=q_A\},\\
        B_0&=\{j\in\mathcal J~|~p_B(E|j)=q_B\},
    \end{aligned}
\end{equation}
correspond to the outcomes that lead Alice (Bob) to assign posterior $q_A$ ($q_B$) to $E$. In analogy with the classical case, we can further define the sets of outcomes
\begin{equation}
    \begin{aligned}
        A_{n+1}&=\{i\in A_n~|~p_A(B_n|i)=1\},\\
        B_{n+1}&=\{j\in B_n~|~p_B(A_n|j)=1\}.
    \end{aligned}
\end{equation}
Given these, Alice's and Bob's posteriors are \textit{common knowledge} when
\begin{equation}
    (i,j)\in \bigcap_{n=0}^\infty A_n\times B_n.
\end{equation}

We can now state the following 
\begin{theorem}[Operational agreement theorem]\label{thm:generalized-aumann}
    Whenever a physical theory allows to compute a joint probability distribution for the outcomes of measurements, if rational agents with a shared prior have common knowledge of their posteriors, those posteriors will be equal.

    Specifically, let $p(i,j,k)$ be the prior probability of the outcomes of three measurements and let $E$ be a subset of the possible values of $k$. Assume that two rational agents, upon learning outcome $i$ and $j$, respectively, assign posterior probability $q_A$ and $q_B$ to $E$. If it is the case that these posteriors are common knowledge, then $q_A=q_B.$
\end{theorem}

\begin{proof}
	Let $A_*=\bigcap_{n} A_n$ and $B_*=\bigcap_{n} B_n$. If either $A_*$ or $B_*$ are empty, there can be no common knowledge of posteriors. Assume, then, that they are both non--empty and that they have nonzero probability. Since the set of measurement outcomes is finite, there is $N$ large enough so that $A_*=A_N$ and $B_*=B_N.$
    Note that 
\begin{equation}
    \begin{aligned}
        q_A = p(E|A_0)=p(E|A_*)=\frac{p(A_*,E)}{p(A_*)},\\
        q_B=p(E|B_0)=p(E|B_*)=\frac{p(B_*,E)}{p(B_*)},
    \end{aligned}
\end{equation}
by definition. 
Since $p(A_{n}|B_{n+1})=p(B_{n}|A_{n+1})=1$, it follows that $p(A_*|B_*)=p(B_*|A_*)=1$, which in turn implies $p(A_*)=p(B_*)$ and $p(A_*,B_*,E)=p(A_*,E)=p(B_*,E)$. Therefore,
\begin{equation}
    \frac{p(A_*,E)}{p(A_*)}=\frac{p(B_*,E)}{p(B_*)},
\end{equation}
from which the result follows.
\end{proof}

Note that this theorem holds trivially for the classical case. Indeed, given the definitions of Section~\ref{sec:aumann}, one can define
\begin{equation}\label{eq:hvm}
p(i,j,k)= \int_{\Pi_i^{(A)}\cap\Pi_j^{(B)}\cap E} p(\omega) \,\mathrm d\omega
\end{equation}
and all the assumptions of the theorem hold.

However the theorem does not specify how $p(i,j,k)$ is obtained: all it needs is a joint probability distribution over three measurement outcomes. Recognising this allows us to extend the applicability of the theorem to scenarios where probabilities over outcomes of measurements are not obtained, like in~\eqref{eq:hvm}, via marginalisation over states of the world.

\begin{figure*}
\includegraphics[width=0.84\linewidth]{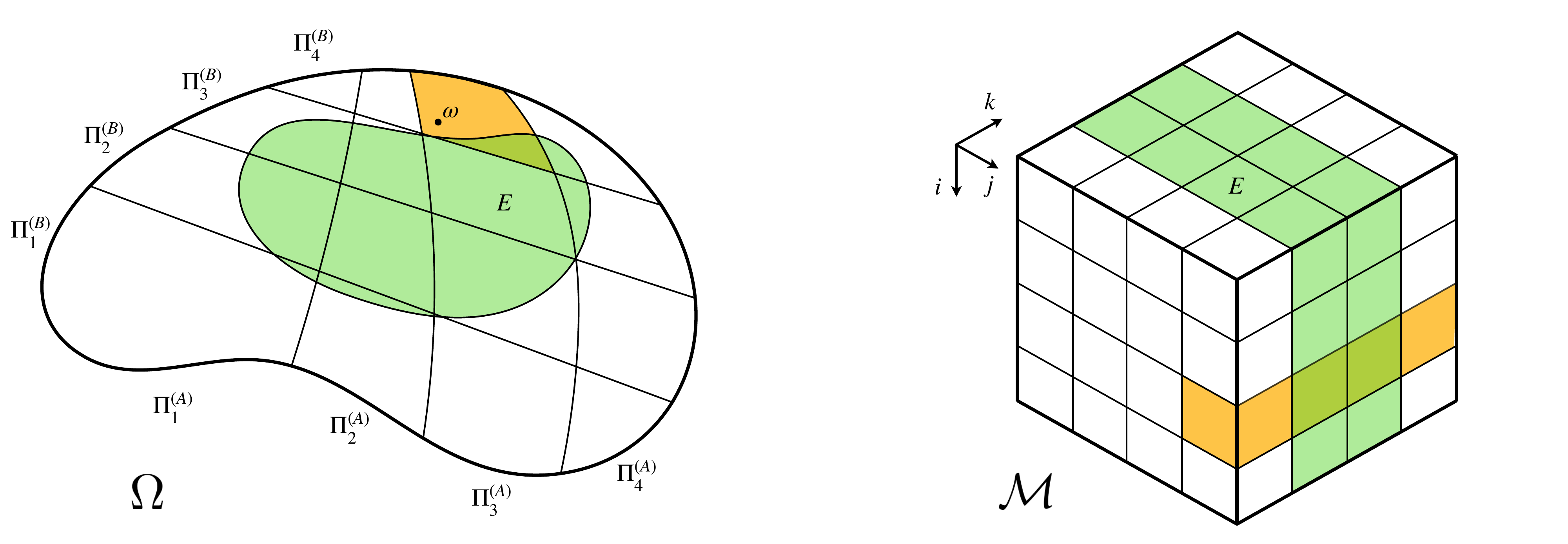} 
    \caption{Left: In the classical agreement theorem, Alice and Bob share a common probability measure $p$ on the state space $\Omega$, and their measurements are described by partitions of $\Omega$; outcomes $i$ and $j$ occur when $\smash{\omega \in \Pi_i^{(A)} \cap \Pi_j^{(B)}}$, and the event $E$ is a subset of $\Omega$. Right: The operational agreement theorem makes no reference to an underlying state of the world and is formulated directly in terms of measurement outcomes.}
    \label{fig:figure}
\end{figure*}

As we will see next, this means that agents cannot ``agree to disagree'' not only when reasoning about classical theory, but in far more exotic settings.

\section{Quantum theory and beyond}

\noindent The theorem we have derived is completely agnostic about the physical origin of the joint probabilities $p(i,j,k)$. This makes it possible to apply it well beyond the classical setting in which Aumann’s theorem is usually formulated. We now show how the same operational reasoning covers arbitrary quantum measurements, including noncommuting measurements and even scenarios in which the order of events is not definite.

Consider the situations in which Alice and Bob perform subsequent measurements on a quantum system, and the event of interest is the outcome of a third subsequent measurement. In this case,
\begin{equation}\label{eq:joint_prob_sequential}
    p(i,j,k)=\tr \mathsf E_k\circ \mathsf B_j\circ \mathsf A_i[\rho],
\end{equation}
where \(\rho\) is the initial state, and \({\mathsf A=\{\mathsf A_i\}_{i\in\mathcal I}}\), \({\mathsf B=\{\mathsf B_j\}_{j\in\mathcal J}}\), and \({\mathsf E=\{\mathsf E_k\}_{k\in\mathcal K}}\) are the instruments describing the three measurements.\footnote{An \textit{instrument} encodes both the measurement probabilities and the state update rule~\cite{davies1970operational,nielsen2010quantum}. It is a collection $\mathsf A=\{\mathsf A_i\},$ of completely positive maps $\mathcal L(\mathcal H)\to\mathcal L(\mathcal H)$ such that $\sum_i\mathsf A_i$ is trace--preserving, ${p(i)=\tr \mathsf A_i[\rho]}$, where the post--measurement state is taken to be $\mathsf A_i[\rho]/p(i)$. For a projective measurement represented by projectors $\{P_i\}$ we would have $P_i[\rho]=P_i\rho P_i$.} The setting closest to Aumann's theorem is the one where the three measurements commute, as recently investigated by~\cite{diaz2025quantum}.
Our theorem, however, does not require \(\mathsf A_i\), \(\mathsf B_j\), and \(\mathsf E_k\) to commute.

As a concrete example, consider the scenario where Alice and Bob perform projective measurements on the same four--dimensional system with \(\mathcal H=\mathbb C^4\). First, Alice performs a measurement on some orthonormal basis \(\{\ket{a_i}\}\), then Bob measures in the basis \(\{\ket{b_j}\}\) defined by
\begin{equation}
\begin{aligned}
    \ket{b_0} &= c_{\theta}\ket{a_0}+s_{\theta}\ket{a_1}, &\ket{b_1} &= -s_{\theta}\ket{a_0}+c_{\theta}\ket{a_1},\\
    \ket{b_2} &= c_{\phi}\ket{a_2}+s_{\phi}\ket{a_3}, &\ket{b_3} &= -s_{\phi}\ket{a_2}+c_{\phi}\ket{a_3},
\end{aligned}
\end{equation}
where \(c_\alpha=\cos\alpha\) and \(s_\alpha=\sin\alpha\). Finally, a third projective measurement is performed, with outcomes ${E_0=\ket{e_0}\!\bra{e_0}}$, and $E_1=\mathbb I-\ket{e_0}\!\bra{e_0},$
with
\begin{equation}
    \ket{e_0}
    =
    \sqrt{q}\ket{b_0}
    +
    \sqrt{q}\ket{b_1}
    +
    \sqrt{r}\ket{b_2}
    +
    \sqrt{1-2q-r}\ket{b_3},
\end{equation}
where \(0<q<1/2\) and \(0<r<1-2q\). For generic values of \((\theta,\phi,q,r)\), the three measurements are noncommuting. The joint probability distribution is given by~\eqref{eq:joint_prob_sequential}
and the posteriors assigned by Alice and Bob to the event \(E\) corresponding to the outcome $k=0$ are
\begin{equation}
    \!\!\!\!\begin{aligned}
        q_A(i)&=\!\big(q,q,c^2_\phi r{+}s^2_\phi(1{-}2q{-}r),\,s^2_\phi r{+}c^2_\phi(1{-}2q{-}r)\big),\\
        q_B(j)&=\!(q,q,r,1{-}2q{-}r).
    \end{aligned}
\end{equation}

This illustrates the statement of Theorem~\ref{thm:generalized-aumann}: whenever there is common knowledge, Alice and Bob must assign the same posterior. For instance, if Alice observes \(i=0\), she can infer that Bob’s outcome is either \(j=0\) or \(j=1\), and hence determine his posterior. The converse holds for Bob and at every higher order of reasoning, implying that there is common knowledge; accordingly, the posteriors coincide. By contrast, if she observes \(i=2\), she does not know Bob's posterior, and their posteriors need not coincide. Once the measurements are fixed, whether common knowledge arises depends on the initial state. For example, if \(\rho\) is supported on the \(\ket{0},\ket{1}\) subspace, then common knowledge always holds; for a generic superposition, it depends on Alice's outcome.

More generally, the agreement theorem will hold for generic measurements $\mathsf A$, $\mathsf B$, and $\mathsf E$, in \textit{whatever} order they are performed.  For example, it will also hold for the case where $\mathsf E$ is performed in between Alice's and Bob's measurement, with
\begin{equation}
    p(i,j,k)=\tr \mathsf {B}_j \circ \mathsf{E}_k\circ \mathsf{A}_i[\rho].
\end{equation}

Indeed, the theorem does not rely on the experiments to be performed in a definite order at all. In the field of indefinite causal order~\cite{oreshkov2012quantum,chiribella2013quantum, araujo2015witnessing}, one asks what is the most general way to assign probability outcomes to measurements in several labs assuming only the local validity of quantum theory in each lab. Alice, Bob, and Eve each receive a quantum system in their lab, apply their instrument, and then send the system out of the lab. It can then be shown that the most general probability assignment to the outcomes of their measurements is
\begin{equation}
    p(i,j,k)=\tr [(\tilde {\mathsf A}_i\otimes\tilde {\mathsf B}_j\otimes \tilde {\mathsf E}_k) W]
\end{equation}
where $\tilde {\mathsf A}_i$ is the positive operator assigned to ${\mathsf A}_i$ via the Choi--Jamiołkoski isomorphism~\cite{jamiolkowski1972linear,choi1975completely} (similarly for $\tilde {\mathsf B}_j$ and $\tilde {\mathsf E}_k$) and $W$, the \textit{process} or \textit{W}-matrix, is some positive operator on the appropriate Hilbert space satisfying a number of consistency requirements~\cite{oreshkov2012quantum}.

A W--matrix can represent situations in which the measurements $\mathsf A$, $\mathsf B$, and $\mathsf E$ happen in some definite but unknown causal order. It can also represent situations in which there is a genuine quantum uncertainty on the order of events, for example, where there is a superposition of Alice's measurement happening before Bob's and Bob's measurement happening before Alice. There are even examples of W--matrices in which no such causal account can be given~\cite{Branciard_2015}. In all of these cases, whenever Alice and Bob have common knowledge of their posteriors, the posteriors will be the same.

Finally, the theorem does not presume the validity of quantum theory at all. For example, it also holds when the probability distribution $p(i,j,k)$ arises from measurements on postquantum systems described by general frameworks such as generalised probabilistic theories~\cite{barrett2007information}, the positive formalism~\cite{oeckl2019local}, or constructor theory~\cite{deutsch2015constructor}. Even if quantum mechanics is eventually superseded by a more fundamental theory, the operational agreement theorem still applies, provided that the theory allows one to compute joint probability distributions for the outcomes of multiple measurements.

\section{Discussion}

\noindent Previous literature on the subject cast doubt on the possibility to extend the agreement theorem beyond quantum theory or even in the case of noncommuting quantum measurements, seemingly contradicting our results.

The seminal Contreras--Tejada \textit{et al.}~\cite{contreras2021observers} derives an agreement theorem for quantum nonsignalling boxes, while showing that postquantum nonsignalling boxes allow for common certainty about different posteriors. There, the events of interest are \textit{counterfactual}: outcomes of measurements that were not performed.\footnote{``If the event $(a = 0, b = 0, x = 0, y = 0)$ is in the sets $A_n$ and $B_n$, for all $n \in \mathbb N$, then, if Alice and Bob \textit{both input $0$} and get output
$0$, we have: Alice assigns probability $q_A$ to $F_B$ (\textit{assuming that Bob input $y = 1$}), Bob is certain that Alice assigns probability $q_A$ to $F_B$, Alice is certain that Bob
is certain that... and so on indefinitely, and vice versa''~\cite[page 5]{contreras2021observers}; emphasis added.} Because of this, they find boxes that allow for ``singular disagreement,'' in which Alice assigns posterior $q_A=1$ to the event of interest while Bob assigns $q_B=0.$ This clearly cannot happen in situations with a well--defined probability distribution over the outcomes, since ${p(k|j)=0}$ for \textit{some} $j,k\in\mathcal J\times \mathcal K$ implies $p(i,j,k)=0$ for all $i\in \mathcal I,$ so whenever Bob assigns posterior $0$ to the event of interest, so must Alice.

More recent work derives an agreement theorem for commuting quantum measurements and showcases an example in which allegedly there is possibility of common knowledge of disagreement if the measurements do not commute~\cite{diaz2025quantum}. There is no contradiction with our theorem, however, as they do not assume the standard postmeasurement update rule to compute posteriors: Alice computes the probability of the event of interest assuming that Bob's system is in a \textit{quantum superposition} of the post--measurement outcomes compatible with her observation. Note that quantum theory would instead tell us that Bob's system is definitely in \textit{one} of the postmeasurement states, albeit unknown to Alice, so Alice should be using a \textit{mixture} rather than a superposition. When all measurements commute, the two update rules yield the same predictions, and common certainty of disagreement is impossible in both frameworks. However, for noncommuting measurements---the regime in which they find common knowledge of disagreement to be possible---the physical interpretation of their nonstandard update rule is unclear.

Our work is aligned and, in a sense, complementary to recent results by Leifer and Duarte~\cite{leifer2026generalising}. 
They extend Aumann's theorem beyond common certainty of \emph{posterior probabilities} to common certainty of \emph{state assignment}, for a quantum or postquantum system $S$: if each agent's state assignment is common knowledge then all agents are assigning the same state to $S.$ The authors remark that all the work in proving the theorem is done by the structure of Bayesian reasoning.

Indeed, the main lesson from the above discussion is that the core of the Agreement theorem does not rely on the existence of an ontology beyond observation, but is actually a statement about the consistency of probabilistic predictions. This can be made explicit by \textit{identifying} the set of outcomes $\mathcal M=\mathcal I\times \mathcal J\times\mathcal K$ with an effective set of possible states of the world $\Omega^\mathrm{op}$. By then setting ${\Pi_i^{(A)}=\{i\}\times\J\times \K}$, 
$\Pi_j^{(B)}=\I\times \{j\}\times\K,$ and
${\mathcal E=\I\times \J\times E}$, we can rely on the mathematical machinery of Aumann's theorem.\footnote{This identification allows us to extend other variants of Aumann's theorem to our most general setting, such as Ref.~\cite{geanakoplos1982we}'s ``dynamical'' version of the theorem, where agents reach agreement by successively disclosing and updating their posteriors; or Ref.~\cite{nielsen1984}, which allows for infinite number of outcomes; or Ref.~\cite{hellman2013almost} which allows from $\varepsilon$--close priors.}

Thus, the reach of the operational agreement theorem is clear: it applies in any setting where such a joint distribution is well defined, independently of how it is obtained. This includes the standard classical case, commuting and noncommuting quantum measurements, scenarios with indefinite causal order described by process matrices, and more general postquantum frameworks. 

Equally clear is where genuinely novel behaviour may emerge: in situations where no joint probability distribution over the relevant outcomes can be meaningfully assigned, and the theorem therefore ceases to apply. This possibility is especially compelling in light of recent discussions of Wigner’s friend and related thought experiments~\cite{brukner2015quantum,bong2020strong,frauchiger2018quantum,schmid2023review,delsanto2025wigners,hausmann2025firewall}. Such scenarios call into question whether outcomes obtained by different observers can always be treated as jointly existing, observer--independent events, and thus provide a natural setting in which to investigate possible departures from the standard agreement framework.

\begin{acknowledgements}
The authors would like to acknowledge the \href{https://www.iqoqi-vienna.at/kefalonia-foundations/participants-kefalonia-2023}{\emph{Kefalonia Foundations 2023}} workshop for fostering the discussions and ideas that led to this work; Fatemeh Bibak for useful discussions during the workshop; and Martin Renner for introducing us to Aumann’s theorem.

This project was funded within the QuantERA II Programme that has received funding from the European Union’s
Horizon 2020 research and innovation programme under Grant Agreement No 101017733, and from the Austrian
Science Fund (FWF), through projects I-6004, ESP2889224, 10.55776/F71, and
10.55776/COE1, as well as Grant No.~I 5384. This work was also made possible through the support of the WOST, WithOut SpaceTime project (\href{https://withoutspacetime.org}{withoutspacetime.org}), supported by Grant ID\#~63683 from the John Templeton Foundation (JTF). The opinions expressed in this work are those of the authors and do not necessarily reflect the views of the John Templeton Foundation.
\end{acknowledgements}

\pagebreak

\bibliography{refs}

\end{document}